\documentclass[11pt]{article}

\usepackage[margin=1in]{geometry}
\usepackage[utf8]{inputenc}
\usepackage{microtype}
\usepackage{amsmath,amsthm,amssymb, color, tikz, mathtools}
\usetikzlibrary{arrows.meta}
\usepackage{caption}

\usepackage{algorithm}
\usepackage{algpseudocode}
\usepackage[pagebackref]{hyperref} 
\hypersetup{colorlinks,linkcolor={blue},citecolor={blue},urlcolor={blue}}
\usepackage{thm-restate}

\newcommand{\zone}{\{0,1\}}
\newcommand{\cbra}[1]{\left\{#1\right\}}
\newcommand{\rbra}[1]{\left(#1\right)}
\newcommand{\eps}{\varepsilon}

\newtheorem{theorem}{Theorem}[section]
\newtheorem{corollary}[theorem]{Corollary}

\newtheorem{lemma}[theorem]{Lemma}

\newtheorem{construction}[theorem]{Construction}

\theoremstyle{definition}
\newtheorem{defi}[theorem]{Definition}

\newcommand{\cA}{\mathcal{A}}
\newcommand{\cR}{\mathcal{R}}
\newcommand{\sD}{\mathsf{D}}
\newcommand{\sR}{\mathsf{R}}
\newcommand{\cT}{\mathcal{T}}

\newcommand{\king}{\textsf{KING}}
\newcommand{\sking}{\textsf{STRONG-KING}}
\newcommand{\eking}{\textsf{EXIST-KING}}

\title{Hardness of Finding Kings and Strong Kings}
\author{Ziad Ismaili Alaoui \thanks{University of Liverpool, United Kingdom. \texttt{\{ziad.ismaili-alaoui, nikhil.mande\}@liverpool.ac.uk}}  \and Nikhil S. Mande \footnotemark[1]}
\date{}

\begin{document}

\maketitle

\begin{abstract}
    A king in a directed graph is a vertex $v$ such that every other vertex is reachable from $v$ via a path of length at most $2$. It is well known that every tournament (a complete graph where each edge has a direction) has at least one king. Our contributions in this work are:
    \begin{itemize}
        \item We show that the query complexity of determining existence of a king in arbitrary $n$-vertex digraphs is $\Theta(n^2)$. This is in stark contrast to the case where the input is a tournament, where Shen, Sheng, and Wu~[SICOMP'03] showed that a king can be found in $O(n^{3/2})$ queries.
        \item In an attempt to increase the ``fairness'' in the definition of tournament winners, Ho and Chang [IPL'03] defined a \emph{strong king} to be a king $k$ such that, for every $v$ that dominates $k$, the number of length-$2$ paths from $k$ to $v$ is strictly larger than the number of length-$2$ paths from $v$ to $k$. We show that the query complexity of finding a strong king in a tournament is $\Theta(n^2)$.
        This answers a question of Biswas, Jayapaul, Raman, and Satti~[DAM'22] in the negative.
    \end{itemize}
    A key component in our proofs is the design of specific tournaments where every vertex is a king, and analyzing certain properties of these tournaments. We feel these constructions and properties are independently interesting and may lead to more interesting results about tournament solutions.
\end{abstract}

\tableofcontents

\section{Introduction}
\label{s:introduction}

A tournament is a complete directed graph. Motivated to find a reasonable notion of a ``most-dominant chicken'' in a flock of chickens (where the ``pecking order'' is represented as a tournament), Landau~\cite{Lan53} introduced the notion of a \emph{king} in a tournament. A king in a tournament is a vertex $k$ such that for every other vertex $v$, either there is an edge from $k$ to $v$, or there exists a $u$ such that there are edges from $k$ to $u$ and from $u$ to $v$. That is, a king is a vertex such that every other vertex is reachable from it via a path of length at most 2. The notion of a king is a natural generalisation of that of a source: a vertex is a source if every other vertex is reachable from it via a path of length 1. While it is easy to see that a tournament need not have a source, it is also not hard to show that every tournament has at least one king~\cite{Lan53, M80}. Soon after Maurer's popular science article about kings in tournaments~\cite{M80}, Reid~\cite{Reid82} showed existence of tournaments in which every vertex is a king. Tournament solutions/winners are also a natural object of study in social choice theory. See, for instance,~\cite{Dey17} and the references therein. The monograph by Moon~\cite{Moon15} sparked a line of research on tournaments, their solutions, and their structural properties.

A natural computational model to study the complexity of finding tournament solutions and related problems is that of \emph{query complexity}. In this setting, an algorithm has access to the input in the form of \emph{queries}: in one step an algorithm may query the direction of the edge between a pair of vertices of its choosing (in case of an undirected graph problem, the algorithm may query the presence/absence of an edge between a pair of vertices of its choosing). The algorithm's goal is to minimise the number of queries made in the worst case, and all other operations are treated as free. There is a rich literature on the study of graph problems in the setting of query complexity, see, for example,~\cite{Ros73, RV76, Yao87, Haj91, CK01, Dey17, MPS23} and the references therein.

The focus of this paper is on finding a king in a digraph. Given that every tournament has a king as mentioned earlier, it is natural to ask: how efficiently can a king in an $n$-vertex tournament be found? 
The study of this question was initiated by Shen, Sheng, and Wu~\cite{SSW03}, who exhibited an algorithm with query complexity $O(n^{3/2})$, and they also showed a non-matching lower bound of $\Omega(n^{4/3})$. To date, these are the best known bounds on the complexity of finding a king in a tournament, and closing the gap remains an intriguing open problem. There have been some very recent works showing better query bounds for variants of the task of finding a king, or studying the complexity of finding kings in models other than deterministic query complexity~\cite{BJRS22, LRT22, MPS23, MPSS24, AGNS24}. Our first main contribution in this paper to show that if we allow for non-tournament inputs, then the query complexity of even determining existence of a king shoots up.

\begin{theorem}\label{thm: main1}
    The randomized query complexity of determining existence of a king in an $n$-vertex digraph is $\Theta(n^2)$.
\end{theorem}

Relevant to our work is the notion of a \emph{strong king} in a tournament. This notion was defined by Ho and Chang~\cite{HC03} in an attempt to increase the ``fairness'' of the notion of a winner in a tournament.
Let $\delta(u, v)$ denote the number of length-2 paths from $u$ to $v$. A king $k$ is said to be \emph{strong} if for all $w$ with $w \rightarrow k$, we have $\delta(k, w) > \delta(w, k)$. To see why this notion increases the fairness of the definition of a winner, consider the following scenario: say we have a tournament in which a vertex $k$ has a single out-neighbour $u$, and all other vertices are in-neighbours of $k$. Further suppose every vertex other than $k$ and $u$ is an out-neighbour of $u$. The vertex $k$ is a king in this tournament, but not necessarily a ``fair'' one, since it ``defeats'' only one vertex. However, it is not hard to see that $k$ is not a strong king.
The study of strong kings in tournaments has also gained some attention~\cite{chen2008existence, muyinda2020non}. Biswas et al.~\cite[Section~7]{BJRS22} asked if one could devise an algorithm to find a strong king in an $n$-vertex tournament with cost $o(n^2)$. Our second main contribution in this paper is to answer this in the negative.

\begin{theorem}\label{thm: main2}
The randomized query complexity of finding a strong king in an $n$-vertex tournament is $\Theta(n^2)$.
\end{theorem}

\subsection{Organization of the Paper}
In Section~\ref{s:preliminaries}, we review basic preliminaries and formalize the computational questions and models of interest. In Section~\ref{s:constructions}, we recall two known constructions of tournaments where every vertex is a king. In Section~\ref{s:lowerbounds}, we show certain properties of these constructions and use these properties to conclude our main results, Theorem~\ref{thm: main1} and Theorem~\ref{thm: main2} (restated formally as Theorem~\ref{thm:randomised-king-existence} and Theorem~\ref{thm:randomised-strong-king}, respectively). Section~\ref{s:lowerbounds} contains the main technical novelty of this paper. Finally, in Section~\ref{s:conclusion}, we recall our main results and discuss potential directions for future work.

\section{Preliminaries}
\label{s:preliminaries}
\subsection{Basic Concepts and Notation}

For a positive integer $n$, let $[n] = \cbra{1, 2, \dots, n}$ and define $[n]_0 = \cbra{0} \cup [n-1]$. We use $S_n$ to denote the set of permutations on the elements in $[n]_0$. For non-negative integers $n$ and $i \in [n]_0$, let $\sigma_i \in S_n$ denote the permutation defined by $\sigma_i(j) = (j + i) \pmod{n}$ for all $j \in [n]_0$. Let $G$ be an $n$-vertex graph with vertex set $[n]_0$, and let $i \in [n]_0$. Define $\sigma(G, i)$ to be the graph obtained by relabelling each vertex $v$ of $G$ by $\sigma_i(v)$.

\begin{defi}[Tournament, Subgraph, Subtournament]
    A \emph{tournament} $T = (V, E)$ is a \emph{complete} directed graph where $V$ is the set of vertices, $E$ is the set of directed edges, and for every edge $(u, v) \in E$, $(v, u) \notin E$. A \emph{subgraph} induced by $G=(V,E)$ on $V' \subseteq V$ is a graph $G[V'] = (V', E')$ such that $E' = E \cap \cbra{(u, v) : u \in V' \textnormal{ and } v \in V'}$. A \emph{subtournament} is analogously defined with respect to tournaments.
\end{defi}

We denote $V(G)$ and $E(G)$ as the vertex and edge sets of $G$, respectively. We often use the notation $(u, v)$ and $u \rightarrow v$ interchangeably. For an edge $e = (u, v)$, we say that $u$ is the \emph{source} of $e$ and $v$ is the \emph{target} of $e$. A vertex $v \in V$ is an \emph{out-neighbor} of $u \in V$ if $(u, v) \in E$, and $v$ is an \emph{in-neighbor} of $u$ if $(v, u) \in E$. We say that a vertex $v$ \emph{dominates} $u$ when $v \to u$. The set of all out-neighbors of a vertex $v$ is denoted by $\Gamma^{+}(v)$, and the set of all in-neighbors of $v$ is denoted by $\Gamma^{-}(v)$. Define $d^+(v) = |\Gamma^+(v)|$ and $d^-(v) = |\Gamma^-(v)|$. A (directed) \emph{path} of length $n$ is an ordered sequence of edges $e_1,e_2, \dots, e_n$ such that the target of $e_i$ is the source of $e_{i+1}$ for all $1 \leq i < n$.

\begin{defi}[King]
\label{def:king}
A \emph{king} in a directed graph is a vertex $v$ such that, for all other vertices $u \neq v$, at least one of the following holds:
\begin{itemize}
    \item $v \to u$ is an edge in the graph, or
    \item there exists a vertex $w$ such that $v \to w$ and $w \to u$ are edges in the graph.
\end{itemize} 
\end{defi}

In other words, a king is a vertex that can reach every other vertex by a path of at most two edges. It is well known~\cite{Lan53} that every tournament has at least one king, and the maximum-out-degree vertex is always a king. A king that dominates all other vertices is called a \emph{source}.

\begin{lemma}[{\cite[Theorem 7]{M80}}]
    \label{lemma:source}
    Let $T = (V, E)$ be an arbitrary tournament. $T$ contains exactly one king if and only if that king is a source.
\end{lemma}

Define $\delta(u,v)$ as the number of length-$2$ paths from vertex $u$ to vertex $v$. A king $k$ is said to be \emph{strong} if, for each vertex $v$ such that $v \to k$, we have $\delta(k,v) > \delta(v, k)$. That is, a king $k$ is said to be strong if, for any vertex $v$ that directly dominates it, there are strictly more length-$2$ paths from $k$ to $v$ than there are from $v$ to $k$. Ho and Chang~\cite[Lemma~1]{HC03} observed that every tournament has at least one strong king. We include a proof for completeness.

\begin{lemma}[{\cite[Lemma~1]{HC03}}]
    \label{lemma:strongkingmaxdegree}
    Let $T = (V, E)$ be an arbitrary tournament, and let $u \in V$ be a maximum-out-degree vertex in $T$. Then, $u$ is a strong king.
\end{lemma}

\begin{proof}
    Towards a contradiction, suppose $u$ is not a strong king. Then, there must exist some vertex $v$ such that $v \to u$ and $\delta(v,u) \geq \delta(u,v)$.
    Define $W = \cbra{w \in V \setminus \cbra{u, v} : u \rightarrow w \textnormal{ and } v \rightarrow w}$. Note that $d^+(u) = \delta(u, v) + |W|$ and $d^+(v) = \delta(v, u) + |W| + 1 > d^+(u)$, which is a contradiction since we assumed $u$ to be a maximum-out-degree vertex.
\end{proof}

A set of vertices $S \subseteq V$ is a \emph{dominating set} in $T = (V, E)$ if, for every $v \in V$, either $v \in S$ or there exists $w \in S$ such that $w \to v$. A \emph{dominating pair} is a dominating set of exactly two vertices.

\subsection{Formal Definitions of Computational Tasks}


We represent an $n$-vertex graph $G=(V=[n]_0,E)$ as a binary string in $\zone^{\alpha_n}$, where ${\alpha_n}=\binom{n}{2}$ if $G$ is a tournament, or ${\alpha_n}=n^2$ otherwise. An element of $[n]_0$ corresponds to the label of a vertex, and there is one variable ($\cbra{i, j}$) per edge (between vertex $i$ and vertex $j$) that defines its direction. In a tournament, we assume $\alpha_n = \binom{n}{2}$ because for each pair of vertices, there is exactly one directed edge between them.
Define the relation $\king_n \subseteq \zone^{\alpha_n} \times [n]_0$ by
    \[
    (G, v) \in \king_n ~\textnormal{if}~\forall u \in [n]_0\setminus \cbra{v}, ~\textnormal{either}~ v \rightarrow u ~\textnormal{or}~\exists w, v \rightarrow w \rightarrow u,
    \]
and define $K_n \subseteq \zone^{\alpha_n}$ as the language of all $n$-vertex graphs $G$ such that $\exists v\in V$, $(G,v)\in \king_n$. In other words, $K$ is the language of all $n$-vertex graphs containing at least one king. Let $\eking_n \subseteq \zone^{\alpha_n} \times \zone$ be the relation such that $(G,1) \in \eking_n$ if and only if $G \in K_n$.\footnote{Note that we could also define $\eking_n$ as a function from $\zone^{\alpha_n} \to \zone$ where $\eking_n(G) = 1$ iff $G$ contains a king, but we choose to define it as a relation for the sake of consistency with other definitions.} Finally, for $v \in V$, define the relation $\sking_n \subseteq \zone^{\alpha_n} \times [n]_0$ by
    \[
    (G, v) \in \sking_n ~\textnormal{if}~\forall u \in [n]_0~\textnormal{such that}~u \to v, \delta(v,u) > \delta(u,v).
    \]

\subsection{Deterministic Query Complexity}

A deterministic decision tree $\cT$ on $m$ variables is a binary tree where the internal nodes are labelled by variables, and leaves are labelled with elements of a set $\cR$. Each internal node has a left child, corresponding to an edge labelled $0$, and a right child, corresponding to an edge labelled $1$. On an input $x \in \zone^m$, $\cT$'s computation traverses a path from root to leaf as follows: At an internal node, the variable associated with that node is \emph{queried}: if the value obtained is $0$, the computation moves to the left child, otherwise it moves to the right child. The output of $\cT$ on input $x$, denoted by $\cT(x)$, is the label of leaf node reached.

We say that a decision tree $\cT$ \emph{computes} the relation $f \subseteq \zone^{m} \times \cR$ if $(x, \cT(x)) \in \cR$ for all $x \in \zone^m$.
The deterministic query complexity of $f$, is 
\[\sD(f) := \min_{\cT : \cT \text{ computes } f} \textnormal{depth}(\cT).\]
That is, $\sD(f)$ represents the minimum depth of any deterministic decision tree that computes $f$.

\subsection{Randomized Query Complexity}

A randomized decision tree $\mathcal{A}$ is a distribution $\mathcal{D}_{\mathcal{A}}$ over deterministic decision trees. On input $x \in \zone^m$, the computation of $\mathcal{A}$ proceeds by first sampling a deterministic decision tree $\cT$ according to $\mathcal{D}_{\mathcal{A}}$, and outputting the label of the leaf reached by $\cT$ on $x$. We say $\mathcal{A}$ computes $f$ to error $\eps$ if for every input $x$, $\Pr[(x, \cA(x)) \in \cR] \geq 1-\eps$. The randomized query complexity of $f \subseteq \zone^m \times \cR$ is defined as follows:
\[\sR_\eps(f) = \min_{\substack{\mathcal{A} \textrm{~computing~} f \\ \textrm{~with error~} \le \textrm{~} \eps}} ~\max_{\cT: \mathcal{D}_{\mathcal{A}}(\cT) > 0}\textnormal{depth}(\cT).\]

When we drop the subscript $\eps$, we assume $\eps = 1/3$.
We use Yao's minimax principle~\cite{Yao77}, stated below in a form convenient for us.
\begin{lemma}[Yao's Minimax Principle]\label{lem:yao}
    For a relation $f \subseteq \zone^m \times \cR$, we have $\sR_\eps(f) \geq k$ if and only if there exists a distribution $\mu : \zone^m \to [0,1]$ such that $\sD_\mu(f) \geq k$.
    Here, $\sD_\mu(f)$ is the minimum depth of a deterministic decision tree that computes $f$ to error at most $\eps$ when inputs are drawn from the distribution $\mu$.
\end{lemma}

\section{Balanced Tournaments}
\label{s:constructions}

A \emph{balanced} tournament is one in which every vertex is a king. Maurer~\cite{M80} showed that a balanced tournament exists for every positive integer $n$, except when $n = 2$ or $n = 4$; this result naturally guarantees the existence of a balanced tournament for every odd $n$.

In this section, we present two distinct constructions of balanced tournaments. In subsequent sections, we analyze key properties of these constructions that help us prove our bounds.
We clarify here that both of these constructions have appeared in prior works, but the properties that we require of them have not been analyzed earlier, to the best of our knowledge.

\begin{construction}[{\cite[Lemma~7]{M80}}]
\label{con:triangle}Let $n$ be an odd positive integer. Define $\Delta_n$ as a tournament on $n$ vertices, labelled by elements of $[n]_0$, recursively constructed as follows:
\begin{itemize}
    \item If $n = 1$, $\Delta_n$ consists of a single vertex labelled $0$.
    \item If $n > 1$, build $\Delta_{n-2}$ as a subtournament induced by $[n-2]_0$. Then direct the edges such that:
    \begin{enumerate}
        \item for each vertex $v$ in $[n-2]_0$, $(n-1) \to v$ and $v \to (n-2)$, and
        \item $(n-2) \to (n-1)$.
    \end{enumerate}
\end{itemize}
\end{construction}

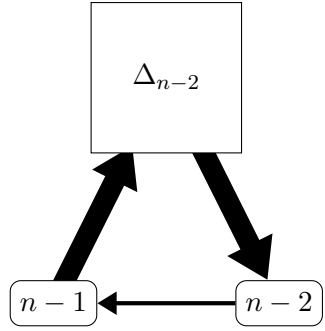
\begin{figure}[!ht]
    \centering
    \begin{tikzpicture}[baseline=(current bounding box.center)]
    \node[shape=rectangle,rounded corners, draw=black, minimum size=6mm] (1) at (3,0){$n-2$};
    \node[shape=circle, minimum size=1em,rounded corners, draw=white, minimum size=2cm] (2) at (1.5,3){};
    \draw[-{Triangle[width=18pt,length=15pt]}, line width=8pt] (0,0) to (2);
    \draw[-{Triangle[width=18pt,length=15pt]}, line width=8pt] (1.5,3) to (1);
    \node[shape=rectangle, minimum size=1em, fill=white,draw=black, minimum size=2cm] (2p) at (1.5,3){$\Delta_{n-2}$};
    \node[shape=rectangle,rounded corners, fill=white, draw=black, minimum size=6mm] (0) at (0,0){$n-1$};
    \draw[>=Triangle, ->, ultra thick] (1) to (0);
\end{tikzpicture}
\vspace{0.5em}
    \caption{Illustration of Construction \ref{con:triangle}.}
    \label{fig:construction-triangle}
\end{figure}

Figure \ref{fig:construction-triangle} illustrates the construction:\footnote{The shape of the figure may clarify the choice of `$\Delta$'.} vertex $n-1$ is dominated by vertex $n-2$ and dominates all vertices in $\Delta_{n-2}$, and vertex $n-2$ dominates vertex $n-1$ and is dominated by all vertices in $\Delta_{n-2}$. For the rest of this paper, we use $\Delta_n$ to denote the $n$-vertex tournament as defined in Construction~\ref{con:triangle}. While the following was already shown by Maurer~\cite[Lemma~7]{M80}, we include a proof for completeness.

\begin{lemma}\label{the:triangle-balanced}
Let $n$ be an odd positive integer. The tournament $\Delta_n$ is balanced.
\end{lemma}

\begin{proof}
    Let us proceed by induction on the number of vertices $n$.

    When $n=1$, $\Delta_n$ is clearly a balanced tournament, as its single vertex trivially satisfies the condition of being a king. Indeed, since there are no other vertices, it vacuously reaches all other vertices by a path of length 1 or 2.

    Now, suppose the lemma holds for some $n$, we show that augmenting the tournament by two vertices following the instructions of Construction \ref{con:triangle} does not violate the invariant; that is, $\Delta_{n+2}$ still remains balanced. By the induction hypothesis, every vertex in $[n]_0$ reaches all other vertices within that set via a path of length 1 or 2. When two new vertices, $n+1$ and $n$, are added:
    \begin{itemize}
        \item vertex $n+1$ dominates all vertices in $[n]_0$;
        \item all vertices in $[n]_0$ dominate $n$; and
        \item vertex $n$ dominates $n+1$.
    \end{itemize} 
    Clearly, none of the vertices in $[n]_0$ loses its status as a king, as each vertex in $[n]_0$ reaches $n$ directly and reaches $n+1$ via $n$. Furthermore, the newly added vertices also satisfy the king condition:
    \begin{itemize}
        \item vertex $n+1$ reaches $n$ via one of the vertices in $[n]_0$, and $n$ reaches $n+1$ directly; and
        \item both $n+1$ and $n$ reach all vertices in $[n]_0$ (either directly or via $n$).
    \end{itemize}
    Thus, the property is preserved and $\Delta_{n+2}$ remains balanced.
\end{proof}

The following construction has appeared in the literature as an example of a \emph{rotational tournament}~\cite[Section~2]{FLRM98}.
\begin{construction}
    \label{con:universal}
    Let $n$ be an odd positive integer. Define $U_n$ as a tournament on $n$ vertices, labelled by elements of $[n]_0$. For each pair of vertices $i, j$ with $i < j$, direct the edge as follows:
    \begin{itemize}
        \item If $i + j$ is odd, direct the edge from $i$ to $j$ ($i \to j$).
        \item Otherwise, direct the edge from $j$ to $i$ ($j \to i$).
    \end{itemize}
\end{construction}

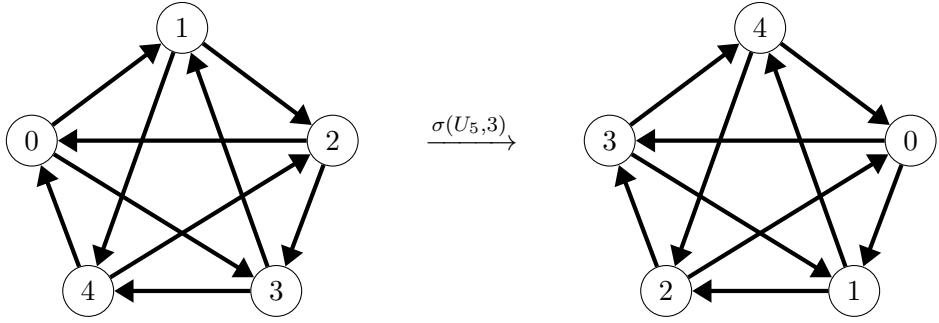
\begin{figure}[!ht]
    \centering
    \begin{tikzpicture}[baseline=(current bounding box.center)]
    \node[shape=circle,rounded corners, draw=black, minimum size=6mm] (0) at (0,0){$0$};
    \node[shape=circle,rounded corners, draw=black, minimum size=6mm] (1) at (2,1.5){$1$};
    \node[shape=circle,rounded corners, draw=black, minimum size=6mm] (2) at (4,0){$2$};
    \node[shape=circle,rounded corners, draw=black, minimum size=6mm] (3) at (3.25,-2){$3$};
    \node[shape=circle,rounded corners, draw=black, minimum size=6mm] (4) at (0.75,-2){$4$};
    \draw[>=Triangle, ->, ultra thick] (0) edge (1) (0) edge (3);
    \draw[>=Triangle, ->, ultra thick] (1) edge (2) (1) edge (4);
    \draw[>=Triangle, ->, ultra thick] (2) edge (3) (2) edge (0);
    \draw[>=Triangle, ->, ultra thick] (3) edge (4) (3) edge (1);
    \draw[>=Triangle, ->, ultra thick] (4) edge (2) (4) edge (0);
\end{tikzpicture}
\hspace{0.5cm}
\begin{tikzpicture}
    \node (0) at (0,0){$\xrightarrow{\sigma(U_5,3)}$};
\end{tikzpicture}
\hspace{0.5cm}
\begin{tikzpicture}[baseline=(current bounding box.center)]
    \node[shape=circle,rounded corners, draw=black, minimum size=6mm] (0) at (0,0){$3$};
    \node[shape=circle,rounded corners, draw=black, minimum size=6mm] (1) at (2,1.5){$4$};
    \node[shape=circle,rounded corners, draw=black, minimum size=6mm] (2) at (4,0){$0$};
    \node[shape=circle,rounded corners, draw=black, minimum size=6mm] (3) at (3.25,-2){$1$};
    \node[shape=circle,rounded corners, draw=black, minimum size=6mm] (4) at (0.75,-2){$2$};
    \draw[>=Triangle, ->, ultra thick] (0) edge (1) (0) edge (3);
    \draw[>=Triangle, ->, ultra thick] (1) edge (2) (1) edge (4);
    \draw[>=Triangle, ->, ultra thick] (2) edge (3) (2) edge (0);
    \draw[>=Triangle, ->, ultra thick] (3) edge (4) (3) edge (1);
    \draw[>=Triangle, ->, ultra thick] (4) edge (2) (4) edge (0);
\end{tikzpicture}
\vspace{0.5em}
    \caption{The tournament $U_5$ and the tournament $\sigma(U_5, 3) \in \mathrm{Aut}(U_5)$ obtained by rotating the labels cyclically clockwise by 3 positions.}
    \label{fig:construction-universal}
\end{figure}

An example is depicted in Figure \ref{fig:construction-universal} when $n=5$. We now show that $U_n$ is balanced for any odd positive integer $n$. Again, for the rest of this paper, we use $U_n$ to denote the $n$-vertex tournament as defined in Construction~\ref{con:universal}.

\begin{lemma}\label{the:universal-balanced}
Let $n$ be an odd positive integer. The tournament $U_n$ is balanced.
\end{lemma}

One way to prove Lemma \ref{the:universal-balanced} is to show that each vertex in $U_n$ (for an odd positive $n$) dominates exactly $\frac{n-1}{2}$ vertices. Thus, given that a maximum-out-degree vertex in a tournament is a king~\cite{Lan53}, it follows that all vertices are kings. For completeness, we include a proof from first principles below.

\begin{proof}[Proof of Lemma~\ref{the:universal-balanced}]
    Observe that each vertex $i$ dominates all lesser vertices of similar parity and all greater vertices of opposite parity. Define the parity function $p(x) = x \pmod{2}$; formally, we have $\Gamma^+(i)=\{j \in [n]:j<i \text{ and } p(j)=p(i)\} \cup \{j \in [n] : j > i \text{ and } p(i) \neq p(j)\}$.

    Consider any vertex $i < n-1$. Since $p(i+1) \neq p(i)$, vertex $i$ dominates $i+1$. Thus, $i+1$ in turn dominates all vertices less than $i$ of opposite parity and all greater vertices of the same parity; formally, $\Gamma^-(i)=\Gamma^+(i+1)$. Hence, $i$ reaches all vertices by a path of at most two edges either through direct domination, or via vertex $i+1$.

    For $i = n-1$, note that $n-1$ is even, so it dominates all even vertices, including vertex $0$. Since vertex $0$ dominates all odd vertices, and $n-1$ dominates $0$, it follows that $n-1$ is also a king.
    Therefore, every vertex in $U_n$ is a king, and the tournament is balanced.
\end{proof}

An interesting property of $U_n$ is its symmetry. Indeed, observe that rotating the tournament (i.e., shifting its labels) clockwise or anti-clockwise results in the same labelled tournament. This is useful as it allows a property of one vertex to be extended to all vertices in $U_n$.

\begin{lemma}[\cite{FLRM98}]
    \label{lm:symmetry}
    Let $n$ be an odd positive integer. For $i\in \mathbb{Z}$, we have $\sigma(U_n,i) \in \mathrm{Aut}(U_n)$.
\end{lemma}

\section{Lower Bounds}
\label{s:lowerbounds}
In this section, we prove our two main results: determining the existence of a king in an arbitrary $n$-vertex digraph (Subsection \ref{ss:kings-digraphs}) and finding a strong king in a tournament (Subsection \ref{ss:strong-kings-tournaments}) both admit a (randomized) query complexity of $\Theta(n^2)$. To do so, we make use of the constructions from Section \ref{s:constructions}.

\subsection{Kings in Directed Graphs}
\label{ss:kings-digraphs}


Below is a property we require of Construction~\ref{con:triangle}. Recall that a 2-element-set of vertices $\cbra{v, w}$ is said to be a \emph{dominating pair} in $T = (V, E)$ if, for every $u \in V \setminus \cbra{v, w}$, there exists $x \in \cbra{v, w}$ such that $x \to u$.

\begin{lemma}\label{lem:triangle-domipairs}

Let $n > 3$ be an odd positive integer, and let $i, j \in [n]_0$ with $i < j$. The set of vertices $\{i, j\}$ is a dominating pair in $\Delta_n$ if and only if $j = n - 1$.
\end{lemma}

\begin{proof}
We show both the directions below.
\begin{itemize}
    \item Let us first show the implication ($\Rightarrow$): that is, any pair $\cbra{i,j}$ where $j = n-1$ is dominating. According to Construction \ref{con:triangle}, and as depicted in Figure \ref{fig:construction-triangle}, vertex $n-1$ dominates all vertices except $n-2$. However, it is easy to see that all other vertices are either $n-2$ or dominate $n-2$. Therefore, pairing $n-1$ with any of those vertices produces a dominating pair.

    \item To show the implication in the other direction ($\Leftarrow$), assume, for the sake of contradiction, that $j \neq n-1$ and $\cbra{i,j}$ is a dominating pair. Then, $i \in [n-2]_0$ and either $j \in [n-2]_0$, or $j = n-2$ (recall that $i < j$). 

    \begin{itemize}
        \item If $j \in [n-2]_0$, then $\cbra{i,j}$ does not dominate $n-1$, a contradiction. 
        \item If $j = n-2$, then $i \in [n-2]_0$ must dominate all other vertices in $[n-2]_0$, and thus must be a source in $\Delta_n[[n-2]_0]$. However, since all (of the at least 3) vertices in $\Delta_n[[n-2]_0]$ are kings by Lemma \ref{the:triangle-balanced}, this contradicts Lemma~\ref{lemma:source}. Thus, the pair $\cbra{i,j}$ is not a dominating pair.
    \end{itemize}
\end{itemize}
This concludes the proof.
\end{proof}

In other words, the only dominating pairs in $\Delta_n$ are $\cbra{0, n-1}, \cbra{1, n-1}, \dots, \cbra{n-2, n-1}$. Using the above property, we show below that the deterministic and randomized query complexity of finding a king in an arbitrary $n$-vertex digraph is $\Theta(n^2)$.

\begin{theorem}
    \label{thm:randomised-king-existence}
    For all positive integers $n$, $\sR(\emph{\eking}_n)=\Theta(n^2)$ for general digraphs.
\end{theorem}

\begin{proof}
    The upper bound is trivial. Assume that $n$ is odd. Let $A = ([n]_0, E_A)$ be an arbitrary balanced tournament on $n$ vertices. Consider the $n$-vertex tournament $\Delta_n$ from Construction~\ref{con:triangle} and relabel all vertex labels by adding $n$ to them. Let the resultant graph be called $A'$. Construct a tournament $C$ on $2n$ vertices labelled in $[2n]_0$ as follows:
    \begin{enumerate}
        \item Consider the disjoint union of $A$ and $A'$, and
        \item add edges from every vertex in $A$ to vertex $(2n-1)$ in $A'$.  
    \end{enumerate}

    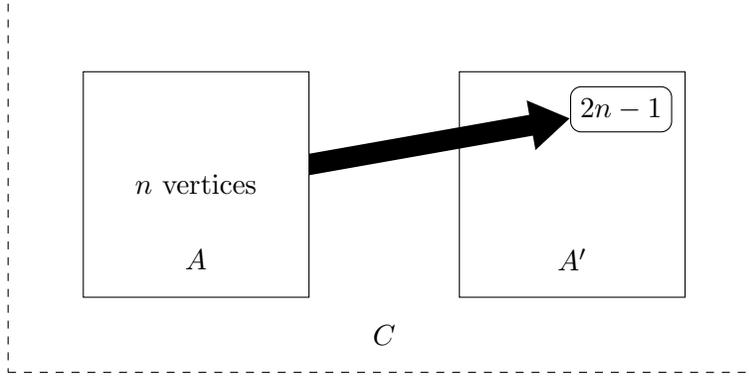
\begin{figure}[!ht]
        \centering
        \begin{tikzpicture}[baseline=(current bounding box.center)]
    \draw[draw=black] (5,0) rectangle ++(3,3);
    \node[] (1) at (6.5,0.5){$A'$};
    \draw[draw=black, dashed] (-1,-1) rectangle ++(10,5);
    \node[] (2) at (4,-0.5){$C$};
    \node[shape=rectangle,rounded corners, draw=black, minimum size=6mm] (3) at (7.15,2.5){$2n-1$};
    \draw[-{Triangle[width=18pt,length=15pt]}, line width=8pt] (1.5,1.5) to (3);
    \draw[draw=black, fill=white] (0,0) rectangle ++(3,3);
    \node[] (0) at (1.5,0.5){$A$};
    \node[] (7) at (1.5,1.5){$n$ vertices};
\end{tikzpicture}
\vspace{0.5em}
        \caption{Construction of $C$ from the proof of Theorem~\ref{thm:randomised-king-existence}. $C[A]$ is a balanced tournament, and $C[A']$ is a relabelling of $\Delta_n$ from Construction~\ref{con:triangle}}
        \label{fig:construction-c-proof}
    \end{figure}
    
    Figure~\ref{fig:construction-c-proof} illustrates the construction of $C$. Clearly, $C$ does not have a king: indeed, no vertex in $A'$ reaches any vertex in $A$, and every vertex in $A$ dominates $(2n - 1)$, which does not dominate $(2n - 2)$ (by construction, see~Construction~\ref{con:triangle}). 

    For $i \in [n]_0$ and $j \in [n-1]_0$, define $C_{i,j}$ as the tournament obtained by adding the edge $i \to (n+j)$ in $C$. Observe that $i$ is a king in $C_{i,j}$, since:
    \begin{itemize}
        \item $i$ 2-step dominates all vertices in $A$ (as $C_{i,j}[A]$ is a balanced tournament), and
        \item $i$ dominates $(2n-1)$ and $n + j$, which form a dominating pair in $C_{i,j}[A']$ by Lemma~\ref{lem:triangle-domipairs}.
    \end{itemize}

    Define the distribution $\mu$ on $2n$-vertex tournaments as follows:
    \begin{align*}
        \mu(T) = \begin{cases}
            1/2 & T = C,\\
            \frac{1}{2n(n-1)} & T = C_{i,j}, \quad \text{for all } i \in [n]_0, j \in [n-1]_0.
        \end{cases}
    \end{align*}
    Equivalently, $\mu$ can be defined by the following sampling procedure: Start with the tournament $C$; then, with probability $1/2$, do nothing, or, with probability $1/2$, add a uniformly random (over $i \in [n]_0$ and $j \in [n-1]_0$) edge $i \rightarrow n+j$. Towards a contradiction using Lemma~\ref{lem:yao}, let $\cA$ be a deterministic algorithm with cost less than $n^2/100$ that decides whether a king exists in digraphs w.r.t.~$\mu$ of error $\leq 1/3$.  

     Consider the computation path taken by $\cA$ that answers all edges consistent with $C$, and answers $0$ (i.e., ``no edge present'') to all queries of the form $i \rightarrow j$ for $i \in [n]$, $j \in \cbra{n, n+1, \dots, 2n - 1}$. We consider two cases.
     \begin{enumerate}
         \item Suppose the output at this leaf of $\cA$ is $1$ (i.e., ``a king exists''). The algorithms errs on the input $C$, which has a $\mu$-mass of 1/2. This is not possible since we assumed $\cA$ to make an error of at most $1/3$ w.r.t.~$\mu$.
         \item Suppose the output at this leaf of $\cA$ is $0$ (i.e., ``no king exists''). Since the number of queries on this path was assumed to be at most $n^2/100$, there must be at least $99n^2/100$ edges of the form $i \to (n + j)$ that have not been queried, where $i \in [n]_0$ and $j \in [n-1]_0$. For each such $i, j$, the graph $C_{i, j}$ also follows this computation path. However, the vertex $i$ is a king in $C_{i, j}$. Hence $\cA$ errs on $C_{i, j}$. Thus, the total error made by $\cA$ is at least $\frac{99n^2}{100} \cdot \frac{1}{2n(n - 1)} > 1/3$ for sufficiently large $n$. This again contradicts our assumption that $\cA$ made error at most $1/3$ w.r.t.~$\mu$, and hence the theorem follows.
     \end{enumerate}
     This concludes the proof.
\end{proof}

As an easy corollary, we obtain the same deterministic lower bound for determining existence of a king in a digraph.
\begin{corollary}
    \label{thm:deterministic-king-existence}
    For all positive integers $n$, $\sD(\emph{\eking}_n)=\Theta(n^2)$ for general digraphs.
\end{corollary}

\subsection{Strong Kings in Tournaments}
\label{ss:strong-kings-tournaments}

Recall that $\delta(v, u)$ denotes the number of length-2 paths from $v$ to $u$ in an underlying digraph/tournament. Recall that a strong king in a tournament is a king $v$ such that, for every vertex $u$ that dominates $v$ (i.e., $u \to v$), we have $\delta(v,u) > \delta(u,v)$. Below, we show that in the construction $U_n$ from Construction~\ref{con:universal}, for every vertex $v \in V(U_n)$ and an in-neighbor $u$ of it, we have $\delta(v, u) - \delta(u, v) = 1$. In other words, there is exactly one more length-$2$ path from $v$ to $u$ than from $u$ to $v$. This observation is crucial to argue that flipping the direction of a single edge used by $v$ to reach $u$ via a length-$2$ path will result in $v$ no longer being a strong king.

\begin{lemma}
    \label{lem:difference-one}
    Let $n$ be an odd positive integer. For each $j \to i$ in $U_n$, we have $\delta(i,j)-\delta(j,i)=1$.
\end{lemma}
\begin{proof}
    Define the parity function $p(x) = x \pmod{2}$. From the proof of Lemma~\ref{the:universal-balanced}, each vertex $i$ in $U_n$ dominates all lesser vertices of the same parity as $i$, and all greater vertices of opposite parity.
    \begin{itemize}
        \item Assume that $j>i$. Since $j \to i$, Construction~\ref{con:universal} implies $p(j) = p(i)$. The vertex $j$ dominates all vertices $v < j$ with $p(v) = p(i)$ and all vertices $w > j$ with $p(w) \neq p(i)$.

    To count length-2 paths from $i$ to $j$, we first show that if $i \to v \to j$, then $i<v<j$. 
    \begin{itemize}
        \item If $v < i$, then $p(v) = p(i)$ (since $i \to v$), but also $p(v) \neq p(j)$ (since $v \to j$ and $v < j$), contradicting $p(i) = p(j)$. 
        \item If $v > j$, then $p(v) \neq p(i)$ (since $i \to v$), but also $p(v) = p(j)$ (since $v \to j$ and $v > j$); again, a contradiction.
    \end{itemize}
        Thus, $v$ must satisfy $i < v < j$, and since $i \to v$, it follows that $p(v) \neq p(i)$. This means   
    $\delta(i, j) = |\{ v : i < v < j \textnormal{ and }  p(v) \neq p(i) \}| = \frac{j - i}{2}$.

    Using nearly the same argument, one can show that  
    $\delta(j, i) = | \{ v : i < v < j \textnormal{ and }  p(v) = p(i) \} | = \frac{j - i}{2} - 1$.

    \item The case where $j < i$ follows from the case where $j > i$ by symmetry (Lemma~\ref{lm:symmetry}). Indeed, let $i'$ and $j'$ be the new labels of $i$ and $j$ in $\sigma(U_n,n-i)$, respectively. Then, $i'=0<j'$, and given that $\sigma(U_n,n-i) \in \mathrm{Aut}(U_n)$, the lemma follows.
    \end{itemize}    
    Therefore, we have $\delta(i,j) - \delta(j,i) = \frac{j - i}{2} - (\frac{j - i}{2} - 1) = 1$.
\end{proof}

For the tournament $U_n$ and a vertex $v$ of it, define $D(v)$ be the set of edges such that, if the direction of exactly one edge in $D(v)$ were flipped, $v$ would no longer be a strong king.

Now, we show that $|D(v)|=\Theta(n^2)$.

\begin{lemma}\label{lem:size-of-dv}
    Let $n$ be an odd positive integer. For the tournament $U_n$ and $D(\cdot)$ as defined above, we have $|D(v)| = (n^2 - 1)/4$ for all $v \in [n]$.
\end{lemma}
\begin{proof}
    We show below that $|D(0)| = (n^2 - 1)/4$. The result for general $v$ immediately follows by symmetry (Lemma \ref{lm:symmetry}), as $|D(0)|=|D(1)|= \cdots = |D(n-1)|$.

    We first consider the effect of flipping an edge connected to 0.
    \begin{itemize}
        \item Let $j \in [n]_0 \setminus \cbra{0}$ be odd. The construction of $U_n$ ensures the direction of the edge between $0$ and $j$ to be $0 \rightarrow j$. Flipping $(0, j)$ increases $\delta(n-1, 0)$ by 1 as a new path $n-1 \rightarrow j \rightarrow 0$ is created, and does not affect $\delta(0, n-1)$.
        \item Let $k \in [n]_0 \setminus \cbra{0}$ be even. The construction of $U_n$ ensures the direction of the edge between $0$ and $k$ to be $k \rightarrow 0$. Flipping the $(k, 0)$ edge cannot increase $\delta(j, 0)$ for any $j$, and cannot decrease $\delta(0, j')$ for any $j'$. Since 0 was a strong king in $U_n$, this means 0 is a strong king in the new tournament with $(k, 0)$ flipped as well.
    \end{itemize}
    The first case above contributes $(n-1)/2$ elements to $D(0)$. Next we consider the effect of flipping an edge between $j$ and $k$ where $j \neq 0$ and $k \neq 0$.
    \begin{itemize}
        \item If $j$ and $k$ are both odd, then neither of them is an in-neighbor of $0$, and hence 0 remains a strong king.
        \item If $j$ and $k$ are both even and $j \rightarrow k$ is flipped, the construction of $U_n$ ensures that $j > k$, and both $j, k$ are in-neighbors of $0$. We have $\delta(j, 0)$ increase by 0 due to the new path $j \rightarrow k \rightarrow 0$, and $\delta(0, j)$ is unchanged. By Lemma~\ref{lem:difference-one}, the new tournament has $\delta(j, 0) = \delta(0, j)$, and hence 0 is no longer a strong king.
        \item If $j$ is even and $k > j$ is odd, the direction of the edge between $j$ and $k$ is $j \rightarrow k$. The vertex $j$ is an in-neighbor of $0$, and $k$ is not. Flipping the edge increases $\delta(0, j)$ by 1 due to the new path $0 \rightarrow j \rightarrow k$, and $\delta(j, 0)$ is unaffected. No other values of $\delta(0, \cdot)$ or $\delta(\cdot, 0)$ are affected for in-neighbors of 0. Thus 0 remains a strong king in this case.
        \item The final case is when $j$ is odd and $k > j$ is even. Recall that $k$ is an in-neighbor of 0. We have $\delta(0, k)$ \emph{reduce} by 1 due to removal of the path $0 \rightarrow j \rightarrow k$ (as the $j \rightarrow k$ edge is now flipped), and $\delta(k, 0)$ is unaffected. By Lemma~\ref{lem:difference-one}, we have $\delta(0, k) = \delta(k, 0)$ in the new tournament, and hence 0 is no longer a strong king.
    \end{itemize}
    The number of edges from the second case above equals the number of ways of choosing two even numbers, none of which is 0: this is $\binom{(n-1)/2}{2}$. The number of edges from the fourth case is the number of choices of an odd vertex and a higher even vertex, which equals $(n-1)/2 + (n-3)/2 + \cdots + 1 = (n^2 - 1)/8$. Combining the above, we have
    \begin{align*}
    |D(0)| & = \frac{n - 1}{2} + \frac{(n-1)(n-3)}{8} + \frac{n^2 - 1}{8} = \frac{n-1}{2} \left[1 + \frac{n-3}{4} + \frac{n+1}{4}\right]\\
    & = \rbra{\frac{n-1}{2}}\rbra{\frac{n + 1}{2}} = \frac{n^2 - 1}{4}.
    \end{align*}
    This concludes the proof.
\end{proof}

We now analyze the randomized query complexity of $\sking_n$.

\begin{theorem}
\label{thm:randomised-strong-king}
    For all positive integers $n$, $\sR(\emph{\sking}_n)=\Theta(n^2)$ for tournaments.
\end{theorem}

\begin{proof}
    The upper bound is trivial.
    Assume that $n$ is odd. For $i, j \in [n]_0$, let $U_n^{i, j}$ denote the tournament that is obtained from $U_n$ by flipping the direction of the edge between $i$ and $j$.
    Consider $\mu$ to be the following distribution on $n$-vertex tournaments:
    \begin{align*}
        \mu(T) = \begin{cases}
            1/2 & T = U_n,\\
            \frac{1}{\binom{n}{2}} & T = U_n^{i, j}, \qquad \textnormal{for all } i, j \in [n]_0.
        \end{cases}
    \end{align*} 
    Towards a contradiction using Lemma~\ref{lem:yao}, let $\cA$ be a deterministic algorithm with cost less than $c \cdot n(n-1)$ (for some constant $c \in (0,1)$ to be fixed later in the proof, it will turn out that any $c < 1/12$ works) computing $\sking_n$ to error less than $1/3$ when inputs are drawn from the distribution $\mu$.
    Consider the computation path of $\cA$ that outputs the answers of all queries consistently with the input $U_n$. Say the output at this leaf of $\cA$ is vertex $v$. The algorithms errs on all inputs of the form $U_n^{i, j}$ where $(i, j) \in D(v)$ (recall that $D(v)$ consists of precisely those edges, that on being flipped in $U_n$ causes $v$ to no longer be a strong king). There are at least $|D(v)| - cn(n-1)$ such edges $(i, j)$ which are unqueried (and thus $U_n^{i, j}$ reaches this leaf).
    The error probability of $\cA$ is at least its error probability on this leaf, which by Lemma~\ref{lem:size-of-dv} is at least $\frac{1}{\binom{n}{2}}\left(\frac{n^2 - 1}{4} - c n(n - 1)\right) > \frac{2}{n(n-1)}\left(\frac{n(n - 1)}{4} - cn(n-1)\right) = \frac{1 - 4c}{2} > 1/3$ for $c < 1/12$. Lemma~\ref{lem:yao} yields the required lower bound.

    The proof can be easily made to work for the case where $n$ is even by assuming that one vertex is a sink, and applying the argument above on the subtournament induced on the remaining vertices (this sink vertex does not affect the proof above in any way).
\end{proof}

A deterministic lower bound follows as an immediate corollary.
\begin{corollary}
    For all positive integers $n$, $\sD(\emph{\sking}_n)=\Theta(n^2)$ for tournaments.
\end{corollary}



\section{Conclusion}
\label{s:conclusion}

In this paper, we have shown the deterministic and randomized query complexities of deciding whether a king exists in an arbitrary $n$-vertex digraph to be $\Theta(n^2)$. We also showed that deterministic and randomized query complexities of finding a strong king in a tournament is $\Theta(n^2)$, answering a question of Biswas et al.~\cite{BJRS22} in the negative. Our proofs relied on showing some key properties of known constructions of balanced tournaments (tournaments where every vertex is a king).
The most interesting question that remains open is to determine the deterministic query complexity of finding a king in an $n$-vertex tournament. The best-known upper and lower bounds are $O(n^{3/2})$ and $\Omega(n^{4/3})$, respectively, and are from over 20 years ago~\cite{SSW03}. Perhaps our analysis of special balanced tournaments in this paper can provide insight towards approaching this general problem.

\bibliography{bibo}

\begin{thebibliography}{CCCW08}

\bibitem[AGNS24]{AGNS24}
Amir Abboud, Tomer Grossman, Moni Naor, and Tomer Solomon.
\newblock From donkeys to kings in tournaments.
\newblock In {\em 32nd Annual European Symposium on Algorithms, {ESA}}, volume 308 of {\em LIPIcs}, pages 3:1--3:14. Schloss Dagstuhl - Leibniz-Zentrum f{\"{u}}r Informatik, 2024.

\bibitem[BJRS22]{BJRS22}
Arindam Biswas, Varunkumar Jayapaul, Venkatesh Raman, and Srinivasa~Rao Satti.
\newblock Finding kings in tournaments.
\newblock {\em Discret. Appl. Math.}, 322:240--252, 2022.

\bibitem[CCCW08]{chen2008existence}
An-Hang Chen, Jou-Ming Chang, Yuwen Cheng, and Yue-Li Wang.
\newblock The existence and uniqueness of strong kings in tournaments.
\newblock {\em Discrete mathematics}, 308(12):2629--2633, 2008.

\bibitem[CK01]{CK01}
Amit Chakrabarti and Subhash Khot.
\newblock Improved lower bounds on the randomized complexity of graph properties.
\newblock In {\em Automata, Languages and Programming: 28th International Colloquium, ICALP 2001 Crete, Greece, July 8--12, 2001 Proceedings 28}, pages 285--296. Springer, 2001.

\bibitem[Dey17]{Dey17}
Palash Dey.
\newblock Query complexity of tournament solutions.
\newblock In {\em Proceedings of the Thirty-First {AAAI} Conference on Artificial Intelligence}, pages 2992--2998. {AAAI} Press, 2017.

\bibitem[FLMR98]{FLRM98}
David~C Fisher, J~Richard Lundgren, Sarah~K Merz, and K~Brooks Reid.
\newblock The domination and competition graphs of a tournament.
\newblock {\em Journal of Graph Theory}, 29(2):103--110, 1998.

\bibitem[Haj91]{Haj91}
P{\'e}ter Hajnal.
\newblock An {$\Omega(n^{4/3})$} lower bound on the randomized complexity of graph properties.
\newblock {\em Combinatorica}, 11:131--143, 1991.

\bibitem[HC03]{HC03}
Ting-Yem Ho and Jou-Ming Chang.
\newblock Sorting a sequence of strong kings in a tournament.
\newblock {\em Information processing letters}, 87(6):317--320, 2003.

\bibitem[Lan53]{Lan53}
HG~Landau.
\newblock On dominance relations and the structure of animal societies: {III} the condition for a score structure.
\newblock {\em The bulletin of mathematical biophysics}, 15:143--148, 1953.

\bibitem[LRT22]{LRT22}
Oded Lachish, Felix Reidl, and Chhaya Trehan.
\newblock When you come at the king you best not miss.
\newblock In {\em 42nd {IARCS} Annual Conference on Foundations of Software Technology and Theoretical Computer Science, {FSTTCS} 2022}, volume 250 of {\em LIPIcs}, pages 25:1--25:12. Schloss Dagstuhl - Leibniz-Zentrum f{\"{u}}r Informatik, 2022.

\bibitem[Mau80]{M80}
Stephen~B Maurer.
\newblock The king chicken theorems.
\newblock {\em Mathematics Magazine}, 53(2):67--80, 1980.

\bibitem[MDR20]{muyinda2020non}
Nathan Muyinda, Bernard {De Baets}, and Shodhan Rao.
\newblock Non-king elimination, intransitive triad interactions, and species coexistence in ecological competition networks.
\newblock {\em Theoretical Ecology}, 13(3):385--397, 2020.

\bibitem[Moo15]{Moon15}
John~W Moon.
\newblock {\em Topics on tournaments in graph theory}.
\newblock Courier Dover Publications, 2015.

\bibitem[MPS23]{MPS23}
Nikhil~S. Mande, Manaswi Paraashar, and Nitin Saurabh.
\newblock Randomized and quantum query complexities of finding a king in a tournament.
\newblock In {\em 43rd {IARCS} Annual Conference on Foundations of Software Technology and Theoretical Computer Science, {FSTTCS}}, volume 284 of {\em LIPIcs}, pages 30:1--30:19. Schloss Dagstuhl - Leibniz-Zentrum f{\"{u}}r Informatik, 2023.

\bibitem[MPSS24]{MPSS24}
Nikhil~S. Mande, Manaswi Paraashar, Swagato Sanyal, and Nitin Saurabh.
\newblock On the communication complexity of finding a king in a tournament.
\newblock In {\em Approximation, Randomization, and Combinatorial Optimization. Algorithms and Techniques, {APPROX/RANDOM}}, volume 317 of {\em LIPIcs}, pages 64:1--64:23. Schloss Dagstuhl - Leibniz-Zentrum f{\"{u}}r Informatik, 2024.

\bibitem[Rei82]{Reid82}
K.B. Reid.
\newblock Every vertex a king.
\newblock {\em Discrete Mathematics}, 38(1):93--98, 1982.

\bibitem[Ros73]{Ros73}
Arnold~L Rosenberg.
\newblock On the time required to recognize properties of graphs: A problem.
\newblock {\em ACM SIGACT News}, 5(4):15--16, 1973.

\bibitem[RV76]{RV76}
Ronald~L Rivest and Jean Vuillemin.
\newblock On recognizing graph properties from adjacency matrices.
\newblock {\em Theoretical Computer Science}, 3(3):371--384, 1976.

\bibitem[SSW03]{SSW03}
Jian Shen, Li~Sheng, and Jie Wu.
\newblock Searching for sorted sequences of kings in tournaments.
\newblock {\em {SIAM} J. Comput.}, 32(5):1201--1209, 2003.

\bibitem[Yao77]{Yao77}
Andrew Chi-Chih Yao.
\newblock Probabilistic computations: Toward a unified measure of complexity.
\newblock In {\em 18th Annual Symposium on Foundations of Computer Science ({SFCS} 1977)}, pages 222--227. IEEE Computer Society, 1977.

\bibitem[Yao87]{Yao87}
Andrew Chi-Chih Yao.
\newblock Lower bounds to randomized algorithms for graph properties.
\newblock In {\em 28th Annual Symposium on Foundations of Computer Science (sfcs 1987)}, pages 393--400. IEEE, 1987.

\end{thebibliography}

\end{document}